\documentclass[11pt]{article}

\usepackage{setspace}
\usepackage{amsmath}
\usepackage{amssymb}
\usepackage{amsthm}
\usepackage{multirow}
\usepackage{color}
\usepackage{longtable}
\usepackage{array}
\usepackage{url}
\allowdisplaybreaks[4]

\newtheorem{theorem}{Theorem}[section]
\newtheorem{definition}{Definition}[section]
\newtheorem{lemma}[theorem]{Lemma}

\newtheorem{problem}{Problem}
\newtheorem{conjecture}[problem]{Conjecture}

\newtheorem{remark}{Remark}[section]

\title{Some new results on permutation polynomials over finite fields}

\author{
Jingxue Ma$^{\text{a}}$, Tao Zhang$^{\text{a}}$, Tao Feng$^{\text{a,c}}$ and Gennian Ge$^{\text{b,c,}}$\thanks{Corresponding author. Email address: majingxue@zju.edu.cn (J. Ma), tzh@zju.edu.cn (T. Zhang), tfeng@zju.edu.cn (T. Feng), gnge@zju.edu.cn (G. Ge)}\\
  \footnotesize $^{\text{a}}$ School of Mathematical Sciences, Zhejiang University, Hangzhou 310027, Zhejiang, China\\
  \footnotesize $^{\text{b}}$ School of Mathematical Sciences, Capital Normal University, Beijing, 100048, China\\
\footnotesize $^{\text{c}}$ Beijing Center for Mathematics and Information
Interdisciplinary Sciences,
Beijing, 100048, China}

\begin{document}

\date{}\maketitle
\begin{abstract}
Permutation polynomials over finite fields constitute an active
research area and have applications in many areas of science and
engineering. In this paper, four classes of monomial complete
permutation polynomials and one class of trinomial complete
permutation polynomials are presented, one of which confirms a conjecture proposed
by Wu et al. (Sci. China Math., 2015, 58, pp. 2081-2094). Furthermore, we give two
classes of permutation trinomial, and  make some progress
on a conjecture about the differential uniformity of power
permutation polynomials proposed by Blondeau et al. (Int. J. Inf. Coding Theory, 2010, 1, pp. 149-170).

\smallskip

\end{abstract}

\section{Introduction}
Let $\mathbb{F}_{p^{n}}$ be a finite field with $p^{n}$ elements,
where $p$ is a prime and $n$ is a positive integer. A polynomial
$f(x)\in\mathbb{F}_{p^{n}}[x]$ is called a permutation polynomial
(PP) if the associated mapping $f:c\mapsto f(c)$ from
$\mathbb{F}_{p^{n}}$ to itself is a bijection. PPs have been intensively studied in recent years due to
their important applications in cryptography, coding theory and combinatorial design theory (see \cite{Ding06,D99,L07,QTTL13,ST05}
and the references therein).  For instance, Ding et al. \cite{Ding06} constructed a family of skew Hadamard difference sets via the Dickson PP of order five, which disproved the longstanding conjecture on skew Hadamard difference sets. Some recent progress on PPs can be found in
\cite{AGW11,DQWYY15,DXYY09,FH12,H15,HX05,H11,H12,LHT13,YD07,YDWP08,ZH12}.

A polynomial $f(x)\in\mathbb{F}_{p^{n}}[x]$ is called a complete
permutation polynomial (CPP) if both $f(x)$ and $f(x)+x$ are
permutations of $\mathbb{F}_{p^{n}}$. These polynomials were
introduced by Niederreiter and Robinson in \cite{NR82}. The simplest
polynomials are monomials, and for a positive integer $d$ and
$\alpha\in\mathbb{F}_{p^{n}}^{*}$, the monomial $\alpha x^{d}$ over
$\mathbb{F}_{p^{n}}$ is a CPP if and only if
$\textup{gcd}(d,q-1)=1$ and $\alpha x^{d}+x$ is a PP. We call such an integer
$d$ a CPP exponent over $\mathbb{F}_{p^{n}}$. Recently, Charpin and
Kyureghyan \cite{CK08} proved that $2^{k}+2$ is a CPP exponent over
$\mathbb{F}_{2^{2k}}$ for odd $k$. In \cite{TZHL}, a class of CPP
exponents over $\mathbb{F}_{2^{n}}$ of Niho type was given. Two new
classes of CPP exponents and a multinomial CPPs were proposed in
\cite{WLHZ}. Further results about CPPs can be found in
\cite{TZH14,WLHZ14,Z}.

In this paper, we construct four classes of monomial CPPs, one class of trinomial CPPs and two classes of permutation trinomials as follows:
\begin{enumerate}
  \item Let $p$ be an odd prime, $r+1=p$ and $d=\frac{p^{rk}-1}{p^{k}-1}+1.$
Then $a^{-1}x^{d}$ is a CPP over $\mathbb{F}_{p^{rk}}$, where $a\in
\mathbb{F}_{p^{rk}}^{*}$ such that $a^{p^{k}-1}=-1$.
  \item Let $n=6k$, where $k$ is a positive integer with $\textup{gcd}(k,3)=1$.
Then $d=2^{4k-1}+2^{2k-1}$ is a CPP exponent over $\mathbb{F}_{2^{n}}$. To be specific,
$a^{-1}x^{d}$ is a CPP over $\mathbb{F}_{2^{n}}$, where $a\in\{\omega^{t(2^{2k}-1)}|0< t\leq 2^{2k}+2^{4k},\ 3\nmid t\}.$
  \item Let $n=4k$. Then $d=(1+2^{2k-1})(1+2^{2k})+1$ is a CPP exponent over $\mathbb{F}_{2^{n}}$.
To be specific, if $a$ is a non-cubic element of $\mathbb{F}_{2^{2k}}^{*}$, then $a^{-1}x^{d}$ is a CPP over $\mathbb{F}_{2^{n}}$.
  \item Let $p$ be an odd prime and $n=4k$. Then $d=\frac{p^{4k}-1}{2}+p^{2k}$ is a CPP exponent over
 $\mathbb{F}_{p^{n}}$. To be specific, $a^{-1}x^{d}$ is a CPP over $\mathbb{F}_{p^{n}}$, where $a\in\{\omega^{t(p^{2k}-1)+\frac{p^{2k}-1}{2}}:\,0\leq t\leq p^{2k}\}.$
  \item For any odd prime $p$,
$f(x)=-x+x^{\frac{p^{2m}+1}{2}}+x^{\frac{p^{2m}+1}{2}p^m}$ is a CPP
over $\mathbb{F}_{p^{3m}}$.
  \item Let $m>1$ be an odd integer, and write $k=\frac{m+1}{2}.$ Then for each $u\in\mathbb{F}_{2^m}^{*}$,
$f(x)=x+u^{2^{k-1}-1}x^{2^k-1}+u^{2^{k-1}}x^{2^k+1}$ is a PP over
$\mathbb{F}_{2^m}$.
  \item Let $m>1$ be an odd integer such that $m=2k-1.$ Then
$f(x)=x+ux^{2^k-1}+u^{2^k}x^{2^m-2^{k+1}+2}, u\in
\mathbb{F}_{2^m}^{*},$ is a PP over
$\mathbb{F}_{2^m}$.
\end{enumerate}

The first class is a conjecture made by Wu et al. \cite{WLHZ}, and our main technique is using the AGW Criterion \cite{AGW11}. By using the additive characters over the underlying finite fields \cite{TZH14}, we give three other new classes (Class 2--Class 4) of monomial CPPs over finite fields. Classes 5, 6 and 7 are explicit constructions of PPs and CPPs through the study of the number of solutions of special equations \cite{DQWYY15,dobb}.

Functions with a low differential uniformity are interesting from
the viewpoint of cryptography as they provide good resistance to
differential attack. In \cite{BCC}, the authors
considered the differential properties of power functions and
proposed the following conjecture.
\begin{conjecture}{\rm\cite{BCC}}\label{conjecture2}
Let $n=2m$ with $m$ odd. Let $F_{d}:x\rightarrow x^d $ be the power PPs over $\mathbb{F}_{2^n}$ defined by the following
values of $d$:
\begin{enumerate}
  \item[(1)] $d=2^m+2^{(m+1)/2}+1,$
  \item[(2)] $d=2^{m+1}+3.$
\end{enumerate}
Then, for these values of $d$, $F_d$ is differentially 8-uniform and
all values $0,2,4,6,8$ appear in its differential spectrum.
\end{conjecture}
To determine the exact value of differential uniformity is a
difficult problem. In this paper, we make some progress towards Conjecture~\ref{conjecture2}, and prove that the differential uniformities of these
polynomials are at most $10$.

This paper is organized as follows. In Section~\ref{preliminaries}, we introduce some basic notations and related results. In Section~\ref{monomial}, four classes of monomial CPPs are given. In Section~\ref{trinomial cpp}, we give a class of trinomial CPPs. Two classes of PPs are presented in Section~\ref{trinomial pp}. Section~\ref{section6} investigates the differential properties of monomial PPs. Section~\ref{conclusion} concludes the paper.

\section{Preliminaries}\label{preliminaries}
The following notations are fixed throughout this paper.
\begin{itemize}
  \item Let $p$ be a prime, $n$ be an integer, and $\mathbb{F}_{p^{n}}$ be the finite field of order $p^{n}$.
  \item Let $\textup{Tr}_{r}^{n}\ :\ \mathbb{F}_{p^{n}}\mapsto\mathbb{F}_{p^{r}}$ be the trace mapping defined by
  $$\textup{Tr}_{r}^{n}(x)=x+x^{p^{r}}+x^{p^{2r}}+\cdots+x^{p^{n-r}},$$
  where $r|n$. For $r=1$, we get the absolute trace function mapping onto the prime field $\mathbb{F}_{p}$, which is denoted by $\textup{Tr}_{n}$.
    \item Let $\textup{N}_{r}^{n}\ :\ \mathbb{F}_{p^{n}}\mapsto\mathbb{F}_{p^{r}}$ be the norm mapping defined by
  $$\textup{N}_{r}^{n}(x)=xx^{p^{r}}x^{p^{2r}}\cdots x^{p^{n-r}},$$
  where $r|n$. For $r=1$, we get the absolute norm function mapping onto the prime field $\mathbb{F}_{p}$, which is denoted by $\textup{N}_{n}$.
  \item Let $\zeta_{p}=\textup{exp}(2\pi\sqrt{-1}/p)$ be a $p$-th root of unity, and $\chi_{n}(x)=\zeta_{p}^{\textup{Tr}_{n}(x)}$ be the canonical additive character on $\mathbb{F}_{p^{n}}$.
\end{itemize}

We first recall a criterion for PPs which can be given by using additive characters over the underlying finite field.

\begin{lemma}{\rm\cite{LN83}}
A mapping $f:\mathbb{F}_{p^{n}}\mapsto\mathbb{F}_{p^{n}}$ is a PP if and only if for every $\alpha\in\mathbb{F}_{p^{n}}^{*}$,
$$\sum_{x\in\mathbb{F}_{p^{n}}}\chi_{n}(\alpha f(x))=0.$$
\end{lemma}

Let $n,r,k$ be integers such that $n=rk$. For any $a\in
\mathbb{F}_{p^n}$, let $a_i=a^{p^{ik}}$, where $0\leq i\leq r-1$. Define
$$h_{a}(x)=x\prod_{i=0}^{r-1}(x+a_i).$$
Then we have the following lemma.

\begin{lemma}{\rm\label{hx}\cite{WLHZ}}
Let $n=rk$, and $d=\frac{p^{rk}-1}{p^{k}-1}+1$. Then $x^{d}+ax\in\mathbb{F}_{p^{n}}[x]$ is a PP over $\mathbb{F}_{p^{n}}$ if and only if $h_{a}(x)\in\mathbb{F}_{p^{k}}[x]$ is a PP over $\mathbb{F}_{p^{k}}$.
\end{lemma}

The following lemmas will also be needed in the following sections.
\begin{lemma}{\rm{\cite{AGW11}}$(${\bf AGW Criterion}$)$}\label{awg}
Let $A, S$ and $\overline{S}$ be finite sets with
$\#S=\#\overline{S}$, and let $f:A\rightarrow A,$ $h:S\rightarrow
\overline{S},$ $\lambda:A\rightarrow S,$ and
$\overline{\lambda}:A\rightarrow \overline{S}$ be maps such that
$\overline{\lambda}\circ f=h\circ\lambda.$ If both $\lambda$ and
$\overline{\lambda}$ are surjective, then the following statements
are equivalent:
\begin{enumerate}
  \item $f$ is bijective;
  \item $h$ is bijective from $S$ to $\overline{S}$ and $f$ is injective on $\lambda^{-1}(s)$ for each $s\in S$.
\end{enumerate}
\end{lemma}
\begin{lemma}\label{cpplem1}{\rm\cite{LN83}}
Let $p$ be an odd prime, and let $m,k$ be positive integers such
that $\frac{m}{\textup{gcd}(m,k)}$ is odd. Then $x^{p^k}+x$ is a
permutation on $\mathbb{F}_{p^m}$.
\end{lemma}
\begin{lemma}\label{lempp1}{\rm\cite{LN83}}
An irreducible polynomial over $\mathbb{F}_q$ of degree $n$ remains irreducible over $\mathbb{F}_{q^k}$ if and only if $\textup{gcd}(k,n)=1$.
\end{lemma}
\section{Four classes of monomial CPPs}\label{monomial}

\subsection{The first class of monomial CPPs}
In this subsection, we will prove the conjecture of Wu et al. \cite{WLHZ}. Before proving it, we establish the following useful lemma.
\begin{lemma}\label{clem}
Let $p$ be an odd prime and $k$ be a positive integer. Then $f(x)=x(x^2-c)^{\frac{p-1}{2}}$ is a PP over $\mathbb{F}_{p^{k}}$, where $c$ is a non-square element in $\mathbb{F}_{p^{k}}$.
\end{lemma}
\begin{proof}
We first show that $x=0$ is the only solution to $f(x)=0$. If $f(x)=0$, then $x=0$ or $(x^2-c)^{\frac{p-1}{2}}=0$. If $(x^2-c)^{\frac{p-1}{2}}=0$, then $c=x^2$, which leads to a contradiction since $c$ is a non-square element. Therefore $f(x)=0$ if and only if $x=0$.

Next, we prove that $f(x)=a$ has a unique nonzero solution for each nonzero $a\in\mathbb{F}_{p^{k}}$. Let $\lambda(x)=x^2-c$, $\overline{\lambda}(x)=x^2$ and $h(x)=(x+c)x^{p-1}$. Then it is easy to see that the following diagram commutes:

\begin{center}
$\begin{array}[c]{ccc}
\mathbb{F}_{p^{k}}^{*}&\stackrel{\lambda}{\rightarrow}&\lambda(\mathbb{F}_{p^{k}}^{*})\\
\downarrow\scriptstyle{f}&&\downarrow\scriptstyle{h}\\
\mathbb{F}_{p^{k}}^{*}&\stackrel{\overline{\lambda}}{\rightarrow}&\overline{\lambda}(\mathbb{F}_{p^{k}}^{*}).
\end{array}$
\end{center}
By Lemma \ref{awg}, it suffices to prove that $h$ is bijective and $f$ is injective on $\lambda^{-1}(s)$ for each $s\in \lambda(\mathbb{F}_{p^{k}}^{*})$. Since for each $s\in \lambda(\mathbb{F}_{p^{k}}^{*})$, $\lambda^{-1}(s)=\{\pm(c+s)^{\frac{1}{2}}\}$, and $f((c+s)^{\frac{1}{2}})\neq f(-(c+s)^{\frac{1}{2}})$, it implies that $f$ is injective on $\lambda^{-1}(s)$ for each $s\in \lambda(\mathbb{F}_{p^{k}}^{*})$.

In the following, we will verify that $h$ is bijective. Since $\#\lambda(\mathbb{F}_{p^{k}}^{*})=\#\overline{\lambda}(\mathbb{F}_{p^{k}}^{*})$, we only need to show that $h$ is injective. For any $b\in \overline{\lambda}(\mathbb{F}_{p^{k}}^{*})$, $b$ is a square element in
$\mathbb{F}_{p^k}^{*}$. We assume
\begin{equation}\label{eq:conj1}
x^p+cx^{p-1}=b
\end{equation}
has at least two distinct solutions. Setting $y=\frac{1}{x}$, the equation
\begin{equation}\label{eq:conj2}
y^p-\frac{c}{b}y-\frac{1}{b}=0
\end{equation}
has at least two distinct solutions. Assume $y_1,$ $y_2$ are two distinct solutions of Eq. (\ref{eq:conj2}). Then $y_1-y_2$ is a root of
$y^p-\frac{c}{b}y=0$, and so must be a root of $y^{p-1}-\frac{c}{b}=0.$ It
follows that $\frac{c}{b}=y_0^{p-1}$ for some $y_0\in \mathbb{F}_{p^k}$,
which is impossible since $\frac{c}{b}$ is a non-square element in $\mathbb{F}_{p^k}$.
Hence Eq. (\ref{eq:conj1}) has at most one solution in $\lambda(\mathbb{F}_{p^{k}}^{*}).$
Therefore, $h(x)$ is bijective. This completes the proof.
\end{proof}

Now we can prove the following result, which is a conjecture made by Wu et al. \cite{WLHZ}.

\begin{theorem}[\cite{WLHZ} Conjecture 4.20]
Let $p$ be an odd prime, $r+1=p$ and $d=\frac{p^{rk}-1}{p^{k}-1}+1.$
Then $a^{-1}x^{d}$ is a CPP over $\mathbb{F}_{p^{rk}}$, where $a\in
\mathbb{F}_{p^{rk}}^{*}$ such that $a^{p^{k}-1}=-1$.
\end{theorem}

\begin{proof}
Since $\textup{gcd}(p^{rk}-1,d)=1$, the monomial $x^d$ is a PP over $\mathbb{F}_{p^{rk}}$.

Note that $a^{p^k-1}=-1$. Then $a^{p^k}=-a$ and $(a^2)^{p^k-1}=1$.
By Lemma \ref{hx}, to prove the conjecture we only need to show that
$h_{a}(x)=x(x^2-a^2)^{\frac{p-1}{2}}$ is a PP over
$\mathbb{F}_{p^{k}}$ for any $k$. Let $c=a^2\in\mathbb{F}_{p^k}.$
Then $c$ is a non-square element in $\mathbb{F}_{p^k}$ since
$a\not\in \mathbb{F}_{p^k}$. Hence the result follows from Lemma
\ref{clem}.
\end{proof}
\subsection{The second class of CPPs}
In this subsection, let $p=2$ and $n=6k$ for some integer $k$
satisfying $\textup{gcd}(k,3)=1$, and $\omega$ be a fixed primitive
element of $\mathbb{F}_{2^{6k}}$. We will show that
$d=2^{4k-1}+2^{2k-1}$ is a CPP exponent over $\mathbb{F}_{2^{6k}}$.
We define the following set
\begin{equation}\label{eqn_S}
S:=\{\omega^{t(2^{2k}-1)}|0< t\leq 2^{2k}+2^{4k},\ 3\nmid t\}.
\end{equation}

\begin{lemma}\label{4.1}
For each $a\in S$,  $\textup{Tr}_{2k}^{6k}(a)\neq1$.
\end{lemma}
\begin{spacing}{1.2}
\begin{proof}

If $a\in S\bigcap\mathbb{F}_{2^{2k}}$, then $\textup{Tr}_{2k}^{6k}(a)=a\neq1$. So we assume that $a\in S\backslash\mathbb{F}_{2^{2k}}$ below.

Since $3|(2^{2k}-1)$, there exists $b\in \mathbb{F}_{2^{6k}}\backslash\mathbb{F}_{2^{2k}}$ such that $b^{3}=a$.
We observe that $\textup{N}_{2k}^{6k}(a)=1$ by the definition of $S$, so  $\eta:=\textup{N}_{2k}^{6k}(b) \in\mathbb{F}_{4}^*$.
Here, $\eta\neq1$, again by the definition of $S$.

Let $B(x)=x^{3}+B_{1}x^{2}+B_{2}x+B_{3}\in\mathbb{F}_{2^{2k}}[x]$ be the
minimal polynomial of $b$ over $\mathbb{F}_{2^{2k}}$. Then $B(x)$ is
irreducible over $\mathbb{F}_{2^{2k}}$, $B_1=\textup{Tr}_{2k}^{6k}(b)$ and $B_{3}=\eta$.
We can directly verify that
\begin{equation}\label{eqn_tr}
\textup{Tr}_{2k}^{6k}(a)=\textup{Tr}_{2k}^{6k}(b^{3})=B_{1}^{3}+B_{1}B_{2}+B_{3}.
\end{equation}
If $B_{1}=0$, then $\textup{Tr}_{2k}^{6k}(a)=B_{3}=\eta\neq1$, and the claim follows.
We assume that $B_{1}\ne 0$ below. Assume to the contrary that
$\textup{Tr}_{2k}^{6k}(a)=1$. Then Eq. \eqref{eqn_tr} gives that
 $B_{2}=\frac{B_{1}^{3}+\eta^{2}}{B_{1}}$, and we have
 \begin{align*}
B(x)&=x^{3}+B_{1}x^{2}+B_{2}x+B_{3} \\
          &=x^{3}+B_{1}x^{2}+\frac{B_{1}^{3}+\eta^{2}}{B_{1}}x+B_{3} \\
          &=(\eta x+B_{1})(\eta^{2}x^{2}+B_{1}x+\frac{\eta}{B_{1}}),
\end{align*}
contradicting to the fact that $B(x)$ is irreducible over $\mathbb{F}_{2^{2k}}$. This completes the proof.
\end{proof}
\end{spacing}

\begin{lemma}\label{4.2}
Fix an integer $k$ with $\textup{gcd}(k,3)=1$. Suppose $n=6k$ and $d=2^{4k-1}+2^{2k-1}$. If $v\in S$, then
$$\sum_{x\in\mathbb{F}_{2^{6k}}}\chi_{6k}(x^{d}+vx)=0.$$
\end{lemma}
\begin{proof}
Let $a$ be a primitive element of $\mathbb{F}_{8}$ with
$a^{3}+a+1=0.$ Since $\textup{gcd}(k,3)=1$, we have
$\mathbb{F}_{2^{6k}}=\mathbb{F}_{2^{2k}}(a)$. For any
$x\in\mathbb{F}_{2^{6k}}$, it can be expressed as
$$x=x_{0}+x_{1}a+x_{2}a^{2},$$
where $x_{0},x_{1},x_{2}\in\mathbb{F}_{2^{2k}}$.

Since $\textup{gcd}(k,3)=1$, we first consider the case $k\equiv1\pmod{3}$, in which $a^{2^{2k}}=a^{4}$. The first step is to compute a direct representation of $\textup{Tr}_{n}(x^{d})$ as a function of $x_{0},x_{1}$ and $x_{2}$. Note that
$\textup{Tr}_{2k}^{6k}(a)=\textup{Tr}_{2k}^{6k}(a^{2})=0$ and $\textup{Tr}_{2k}^{6k}(1)=1$. A routine computation shows that
$$\textup{Tr}_{6k}(x^{d})=\textup{Tr}_{2k}(x_{0}+x_{1}+x_{2}+x_{1}x_{2}).$$
Next, putting
$$v=v_{0}+v_{1}a+v_{2}a^{2}$$
with $v_{0},v_{1},v_{2} \in \mathbb{F}_{2^{2k}}$, we find that
$$\textup{Tr}_{6k}(vx)=\textup{Tr}_{2k}(v_{0}x_{0}+v_{1}x_{2}+v_{2}x_{1}).$$
Consequently,
$$\sum_{x\in\mathbb{F}_{2^{6k}}}\chi_{6k}(x^{d}+vx)=\sum_{x_{1},x_{2}\in\mathbb{F}_{2^{2k}}}\chi_{2k}(x_{1}+x_{2}+x_{1}x_{2}+v_{1}x_{2}+v_{2}x_{1})
\sum_{x_{0}\in\mathbb{F}_{2^{2k}}}\chi_{2k}(x_{0}+v_{0}x_{0}).$$

From Lemma \ref{4.1} we get $v_{0}\neq1$. Therefore, $\sum_{x\in\mathbb{F}_{2^{6k}}}\chi_{6k}(x^{d}+vx)=0.$

For the remaining case $k\equiv2\pmod{3}$, a similar discussion leads to $\sum_{x\in\mathbb{F}_{2^{6k}}}\chi_{6k}(x^{d}+vx)=0.$
\end{proof}

\begin{theorem}
Let $n=6k$, where $k$ is a positive integer with $\textup{gcd}(k,3)=1$.
Then $d=2^{4k-1}+2^{2k-1}$ is a CPP exponent over $\mathbb{F}_{2^{n}}$. To be specific,
 if $a\in S$ with $S$ defined as in \eqref{eqn_S}, then $a^{-1}x^{d}$ is a CPP over $\mathbb{F}_{2^{n}}$.
\end{theorem}
\begin{proof}
Since $\textup{gcd}(d,2^{6k}-1)=1$, we have $x^{d}$ is a PP over
$\mathbb{F}_{2^{6k}}$. In what follows we prove that $x^{d}+ax$ is
also a PP over $\mathbb{F}_{2^{6k}}$. We only need to prove that for
each $\alpha\in\mathbb{F}_{2^{6k}}^{*}$,
 $$\sum_{x\in\mathbb{F}_{2^{6k}}}\chi_{6k}(\alpha(x^{d}+ax))=0,$$
 where $a\in S$. Since $\textup{gcd}(d,2^{6k}-1)=1$, each nonzero $\alpha\in\mathbb{F}_{p^{6k}}$ can be written as $\alpha=\beta^{d}$ for a unique $\beta\in\mathbb{F}_{2^{6k}}^{*}$, and we have
 \begin{align*}
\sum_{x\in\mathbb{F}_{2^{6k}}}\chi_{6k}(\alpha(x^{d}+ax))&=\sum_{x\in\mathbb{F}_{2^{6k}}}\chi_{6k}((\beta x)^{d}+\beta^{d-1}a\beta x) \\
          &=\sum_{x\in\mathbb{F}_{2^{6k}}}\chi_{6k}(x^{d}+\beta^{d-1}ax) \\
          &=\sum_{x\in\mathbb{F}_{2^{6k}}}\chi_{6k}(x^{d}+\beta^{2^{4k-1}+2^{2k-1}-1}ax).
\end{align*}
Since $\beta^{2^{4k-1}+2^{2k-1}-1}a\in S$, it follows from Lemma
\ref{4.2} that for each $\alpha\in\mathbb{F}_{2^{6k}}^{*}$, we have
 $$\sum_{x\in\mathbb{F}_{2^{6k}}}\chi_{6k}(\alpha(x^{d}+ax))=0.$$
  This completes the proof.
\end{proof}
\subsection{The third class of monomial CPPs}
In this subsection, let $p=2$ and $n=4k$. We will use an analysis similar to that of the previous subsection to show that $d=(1+2^{2k-1})(1+2^{2k})+1$ is a CPP exponent over $\mathbb{F}_{2^{4k}}$.

\begin{lemma}\label{3.1}
If $v$ is a non-cubic element of $\mathbb{F}_{2^{2k}}^{*}$, then
$$\sum_{x\in\mathbb{F}_{2^{4k}}}\chi_{4k}(x^{(1+2^{2k-1})(1+2^{2k})+1}+vx)=0.$$
\end{lemma}
\begin{proof}
Using polar coordinate representation, every nonzero element $x$ of $\mathbb{F}_{2^{4k}}$ can be uniquely represented as $x=yz$ with $y\in U$ and $z\in \mathbb{F}_{2^{2k}}^{*}$, where $U=\{\lambda\in\mathbb{F}_{2^{4k}}| \lambda^{2^{2k}+1}=1\}.$ Then
 \begin{align*}
\sum_{x\in\mathbb{F}_{2^{4k}}}\chi_{4k}(x^{(1+2^{2k-1})(1+2^{2k})+1}+vx)&=1+\sum_{x\in\mathbb{F}_{2^{4k}}^{*}}\chi_{4k}(x^{(1+2^{2k-1})(1+2^{2k})+1}+vx) \\
          &=1+\sum_{y\in U}\sum_{z\in\mathbb{F}_{2^{2k}}^{*}}\chi_{4k}((yz)^{(1+2^{2k-1})(1+2^{2k})+1}+vyz)\\
          &=-2^{2k}+\sum_{y\in U}\sum_{z\in\mathbb{F}_{2^{2k}}}\chi_{4k}(yz^{4}+vyz)\\
          &=-2^{2k}+\sum_{y\in U}\sum_{z\in\mathbb{F}_{2^{2k}}}\chi_{2k}(\textup{Tr}_{2k}^{4k}(yz^{4}+vyz))\\
          &=-2^{2k}+\sum_{y\in U}\sum_{z\in\mathbb{F}_{2^{2k}}}\chi_{2k}((y+y^{2^{2k}})z^{4}+(y+y^{2^{2k}})vz)\\
          &=-2^{2k}+\sum_{y\in U}\sum_{z\in\mathbb{F}_{2^{2k}}}\chi_{2k}(z^{4}(y+y^{2^{2k}}+y^{4}v^{4}+y^{2^{2k+2}}v^{4}))\\
          &=(N(v)-1)2^{2k},
\end{align*}
where $N(v)$ denotes the number of $y$'s in $U$ such that $y+y^{2^{2k}}+y^{4}v^{4}+y^{2^{2k+2}}v^{4}=0$, which is equivalent to
$$y+y^{-1}+y^{4}v^{4}+y^{-4}v^{4}=0,$$
that is,
$$(y+y^{-1})[1+v^{4}(y+y^{-1})^{3}]=0.$$
Since $v$ is a non-cubic element of $\mathbb{F}_{2^{2k}}^{*}$, we
get $1+v^{4}(y+y^{-1})^{3}\neq0$. Hence $y=1$ is the unique root.
Thus $N(v)=1$, and this completes the proof.
\end{proof}

\begin{theorem}
Let $n=4k$. Then $d=(1+2^{2k-1})(1+2^{2k})+1$ is a CPP exponent over $\mathbb{F}_{2^{n}}$.
To be specific, if $a$ is a non-cubic element of $\mathbb{F}_{2^{2k}}^{*}$, then $a^{-1}x^{d}$ is a CPP over $\mathbb{F}_{2^{n}}$.
\end{theorem}
\begin{proof}
It can be verified that $\textup{gcd}(d,2^{4k}-1)=1$. Thus, it suffices to prove that $x^{d}+ax$ is a PP for a non-cubic $a\in\mathbb{F}_{2^{2k}}^{*}$. Since each nonzero $\alpha\in\mathbb{F}_{2^{4k}}$ can be written as $\alpha=\beta^{d}$ for a unique $\beta\in\mathbb{F}_{2^{4k}}^{*}$, we have
 \begin{align*}
\sum_{x\in\mathbb{F}_{2^{4k}}}\chi_{4k}(\alpha(x^{d}+ax))&=\sum_{x\in\mathbb{F}_{2^{4k}}}\chi_{4k}((\beta x)^{d}+\beta^{d-1}a\beta x) \\
          &=\sum_{x\in\mathbb{F}_{2^{4k}}}\chi_{4k}(x^{d}+\beta^{d-1}ax) \\
          &=\sum_{x\in\mathbb{F}_{2^{4k}}}\chi_{4k}(x^{d}+\beta^{(1+2^{2k-1})(1+2^{2k})}ax).
\end{align*}
Note that $\beta^{(1+2^{2k-1})(1+2^{2k})}a$ is also a non-cubic element of
$\mathbb{F}_{2^{2k}}^{*}$. For each $\alpha\in\mathbb{F}_{2^{4k}}^{*}$, we have
 $$\sum_{x\in\mathbb{F}_{2^{4k}}}\chi_{4k}(\alpha(x^{d}+ax))=0,$$
by Lemma \ref{3.1}. This completes the proof.
\end{proof}

\subsection{The fourth class of monomial CPPs}

In this subsection, we study the fourth class of monomial CPPs,  where $p$ is an odd prime, $n=4k$ and $d=\frac{p^{4k}-1}{2}+p^{2k}$. Let $\omega$ be a fixed primitive element of $\mathbb{F}_{p^{n}}$. Denote the conjugate of $x$ over $\mathbb{F}_{p^{n}}$ by $\overline{x}$, i.e. $\overline{x}=x^{p^{2k}}$. We also define the set
$S$ as follows:
\begin{equation}\label{eqn_S34}
  S:=\{\omega^{t(p^{2k}-1)+\frac{p^{2k}-1}{2}}:\,0\leq t\leq p^{2k}\}.
\end{equation}

We first recall two lemmas.
\begin{lemma}\rm{\cite{H76}\label{2.1}}
Let $p$ be an odd prime and $d|p^{n}-1$. Let $s$ be the least positive integer
such that $d|p^{s}+1$. For each $0\leq j<d$, define the set
\[
 C_{j}:=\{\omega^{di+j}\in\mathbb{F}_{p^{n}}^{*}|0\leq i<\frac{p^{n}-1}{d}\}.
\]
\begin{enumerate}
  \item In the case $d$ is an even integer, and both $(p^{s}+1)/d$ and $d/2s$ are odd integers,
  we have
   \[\sum_{x\in\mathbb{F}_{p^{n}}}\chi_{n}(ax^{d})=\begin{cases}p^{n};&\textup{ if } a=0,\\
(-1)^{\frac{n}{2s}+1}(d-1)p^{\frac{n}{2}};&\textup{ if }a\in C_{\frac{d}{2}},\\
(-1)^{\frac{n}{2s}}p^{\frac{n}{2}};&\textup{ if }a\not\in C_{\frac{d}{2}}.\end{cases}\]
  \item In all the other cases, we get
     \[\sum_{x\in\mathbb{F}_{p^{n}}}\chi_{n}(ax^{d})=\begin{cases}p^{n};&\textup{ if } a=0,\\
(-1)^{\frac{n}{2s}+1}(d-1)p^{\frac{n}{2}};&\textup{ if }a\in C_{0},\\
(-1)^{\frac{n}{2s}}p^{\frac{n}{2}};&\textup{ if }a\not\in C_{0}.\end{cases}\]
\end{enumerate}

\end{lemma}

\begin{lemma}\rm{\cite{H76}\label{2.2}}
Let $d$ be an integer with $\textup{gcd}(d,p^{n}-1)=1$. Suppose that there exists an integer $i$
 such that $0\leq i<n$ and $(d-p^{i})|(p^{n}-1)$. Choose an integer $N$ such that $(d-p^{i})N\equiv0\pmod{p^{n}-1}$. Then
$$\sum_{x\in\mathbb{F}_{p^{n}}}\chi_{n}(x^{d}+ax)=\frac{1}{N}\sum_{j=0}^{N-1}\sum_{y\in\mathbb{F}_{p^{n}}}\chi_{n}(y^{N}(a\omega^{j}+\omega^{djp^{-i}})).$$
\end{lemma}

As a preparation, we have the following lemmas.
\begin{lemma}\label{2.3}
If $a\in S$ with $S$ defined as in \eqref{eqn_S34}, then $\frac{a+1}{a-1}=\omega^{2s}$ for some integer $s$.
\end{lemma}
\begin{proof}
Assume to the contrary that $\frac{a+1}{a-1}=\omega^{2s+1}$ for some integer $s$. Since $a\overline{a}=-1$, we have
$$\frac{a+1}{a-1}=\frac{a-a\overline{a}}{a+a\overline{a}}=\frac{1-\overline{a}}{1+\overline{a}}=(\frac{1-a}{1+a})^{p^{2k}}=-\omega^{-(2s+1)p^{2k}}.$$
It follows that
$$\omega^{(2s+1)(p^{2k}+1)}=-1=\omega^{(p^{2k}+1)\frac{p^{2k}-1}{2}},$$
which is a contradiction. So we have $\frac{a+1}{a-1}=\omega^{2s}$
for some integer $s$.
\end{proof}

\begin{lemma}\label{2.4}
Let $p$ be an odd prime, $n=4k$ and $d=\frac{p^{4k}-1}{2}+p^{2k}$. If $a\in S$ with $S$ defined as
in \eqref{eqn_S34}, then $\sum_{x\in\mathbb{F}_{p^{n}}}\chi_{n}(x^{d}+ax)=0$.
\end{lemma}

\begin{proof}
By Lemma \ref{2.2}, we have
 \begin{align*}
2\sum_{x\in\mathbb{F}_{p^{n}}}\chi_{n}(x^{d}+ax)&=\sum_{y\in\mathbb{F}_{p^{n}}}\chi_{n}(y^{2}(a+1))+\sum_{y\in\mathbb{F}_{p^{n}}}\chi_{n}(y^{2}(a\omega+\omega^{dp^{2k}})) \\
          &=\sum_{y\in\mathbb{F}_{p^{n}}}\chi_{n}(y^{2}(a+1))+\sum_{y\in\mathbb{F}_{p^{n}}}\chi_{n}(y^{2}((a-1)\omega)).
\end{align*}
From Lemma \ref{2.3}, $a-1,\ a+1\in C_{0}$ or $a-1,\ a+1\in C_{1}$. Then a direct application of Lemma \ref{2.1} shows $\sum_{x\in\mathbb{F}_{p^{n}}}\chi_{n}(x^{d}+ax)=0$.
\end{proof}

We now have the following theorem.

\begin{theorem}
Let $p$ be an odd prime and $n=4k$. Then $d=\frac{p^{4k}-1}{2}+p^{2k}$ is a CPP exponent over
 $\mathbb{F}_{p^{n}}$. To be specific, if $a\in S$ with $S$ defined as in \eqref{eqn_S34}, then $a^{-1}x^{d}$ is a CPP over $\mathbb{F}_{p^{n}}$.
\end{theorem}
\begin{proof}
Since $\textup{gcd}(d,p^{n}-1)=1$, for each $a\in S$, the monomial $a^{-1}x^{d}$ is a PP over $\mathbb{F}_{p^{n}}$. To finish the proof, it suffices to prove that $x^{d}+ax$ permutes $\mathbb{F}_{p^{n}}$.

The fact $\textup{gcd}(d,p^{n}-1)=1$ shows that each nonzero element $\alpha\in\mathbb{F}_{p^{n}}$ can be represented as $\alpha=\beta^{d}$ for a unique $\beta\in\mathbb{F}_{p^{n}}^{*}$. Then
 \begin{align*}
\sum_{x\in\mathbb{F}_{p^{n}}}\chi_{n}(\alpha(x^{d}+ax))&=\sum_{x\in\mathbb{F}_{p^{n}}}\chi_{n}((\beta x)^{d}+\beta^{d-1}a\beta x) \\
          &=\sum_{x\in\mathbb{F}_{p^{n}}}\chi_{n}(x^{d}+\beta^{d-1}ax)\\
          &=\sum_{x\in\mathbb{F}_{p^{n}}}\chi_{n}(x^{d}+\beta^{(\frac{p^{2k}+1}{2}+1)(p^{2k}-1)}ax).
\end{align*}
Since $\beta^{(\frac{p^{2k}+1}{2}+1)(p^{2k}-1)}a\in S$, it follows
from Lemma \ref{2.4} that for each
$\alpha\in\mathbb{F}_{p^{n}}^{*}$, we have
 $$\sum_{x\in\mathbb{F}_{p^{n}}}\chi_{n}(\alpha(x^{d}+ax))=0.$$
  This completes the proof.
\end{proof}

\section{A class of trinomial CPPs}\label{trinomial cpp}
In this section, we consider a class of trinomial polynomials and show that they are CPPs.

\begin{theorem}
For any odd prime $p$,
$f(x)=-x+x^{\frac{p^{2m}+1}{2}}+x^{\frac{p^{2m}+1}{2}p^m}$ is a CPP
over $\mathbb{F}_{p^{3m}}$.
\end{theorem}
\begin{proof}
Note that
$$f(x)+x=x^{\frac{p^{2m}+1}{2}}+x^{\frac{p^{3m}+p^m}{2}}=g(x^{\frac{p^{2m}+1}{2}})$$
with $g(x)=x+x^{p^m}$. It follows that $f(x)+x$ is a PP over
$\mathbb{F}_{p^{3m}}$ if and only if $g(x)$ is so, since
$\textup{gcd}(\frac{p^{2m}+1}{2},p^{3m}-1)=1.$ By Lemma
\ref{cpplem1},   $g(x)$ is a PP over $\mathbb{F}_{p^{3m}}$, so $f(x)+x$ is a PP.

In the following, we show that $f(x)$ is a PP over $\mathbb{F}_{p^{3m}}$.
Write $h(x):=x+x^{p^m}-x^{1+p^m-p^{2m}}$, so that $f(x)=h(x^{\frac{p^{2m}+1}{2}})$. Since
$\textup{gcd}(\frac{p^{2m}+1}{2},p^{3m}-1)=1$, $f(x)$ is a PP over
$\mathbb{F}_{p^{3m}}$ if and only if $h(x)$ is a PP over
$\mathbb{F}_{p^{3m}}.$ Note that $h(0)=0$ and for $x\ne 0$,
$$h(x)=\frac{x^{1+p^{2m}}+x^{p^m+p^{2m}}-x^{1+p^m}}{x^{p^{2m}}}.$$

We first prove that $h(x)=0$ if and only if $x=0$. Suppose to the contrary that
$h(x)=0$ for some $x\ne 0$, that is,
$$x^{1+p^{2m}}+x^{p^m+p^{2m}}-x^{1+p^m}=0.$$
Raising the above equation to the $p^m$-th power and the
$p^{2m}$-th power respectively, we obtain
\begin{align*}
x^{1+p^m}+x^{1+p^{2m}}-x^{p^m+p^{2m}}&=0,\\
x^{p^m+p^{2m}}+x^{1+p^m}-x^{1+p^{2m}}&=0.
\end{align*}
Adding the above two equations together, we deduce that $2x^{1+p^m}=0$. So $x=0$,
contradicting to the assumption that $x\in\mathbb{F}_{p^{3m}}^{*}$. Thus $h(x)=0$ if and only if $x=0$.

Next, we prove that $h(x)=a$ has at most one nonzero solution for each $a\in\mathbb{F}_{p^{3m}}^{*}$, that is
\begin{equation}\label{eq:cpp2}
x^{1+p^{2m}}+x^{p^m+p^{2m}}-x^{1+p^m}=ax^{p^{2m}}
\end{equation}
has at most one nonzero solution for each $a\in
\mathbb{F}_{p^{3m}}^{*}.$ Raising both sides of Eq. \eqref{eq:cpp2} to the power of $p^m$ and $p^{2m}$ respectively, we
get
\begin{align*}
x^{1+p^m}+x^{1+p^{2m}}-x^{p^m+p^{2m}}&=a^{p^m}x,\\
x^{p^m+p^{2m}}+x^{1+p^m}-x^{1+p^{2m}}&=a^{p^{2m}}x^{p^m}.
\end{align*}
Adding the above two equations together, we obtain
\begin{equation}\label{eq:cpp5}
2x^{1+p^m}=a^{p^m}x+a^{p^{2m}}x^{p^m}.
\end{equation}
Because $x\neq0$, we have $2x^{p^m}=a^{p^m}+a^{p^{2m}}x^{p^m-1}$.
Setting $y=\frac{1}{x}$, we get
\begin{equation}\label{eq:cpp6}
y^{p^m}+a^{p^{2m}-p^m}y-2a^{-p^m}=0.
\end{equation}
If Eq. \eqref{eq:cpp6} has at least two distinct nonzero
solutions $y_1$, $y_2$ in $\mathbb{F}_{p^{3m}}$ we get $y_1-y_2\in
\mathbb{F}_{p^{3m}}$ is a root of $y^{p^m}+a^{p^{2m}-p^m}y=0$, and
so a root of $y^{p^m-1}+a^{p^{2m}-p^m}=0,$ contradicting to the fact
that $y^{p^m-1}+a^{p^{2m}-p^m}=0$ has no solution in
$\mathbb{F}_{p^{3m}}^{*}.$ Therefore, Eq. \eqref{eq:cpp6} has
at most one nonzero solution in $\mathbb{F}_{p^{3m}}.$ Hence
$h(x)=a$ has at most one nonzero solution for each nonzero $a\in
\mathbb{F}_{p^{3m}}$.

To sum up, we have shown that $h(x)$ is a PP, and thus $f(x)$ is a PP. The proof is now complete.
\end{proof}

\section{Two classes of trinomial PPs}\label{trinomial pp}
It looks difficult to give a simple characterization of trinomial PPs over finite fields. In \cite{DQWYY15}, the authors use different tricks including the multivariate method introduced by Dobbertin \cite{dobb,D99} to construct several classes of trinomial PPs. In this section, we construct two classes of trinomial PPs over $\mathbb{F}_{2^m}$ by similar techniques.
\begin{theorem}
Let $m>1$ be an odd integer, and write $k=\frac{m+1}{2}.$ Then for each $u\in\mathbb{F}_{2^m}^{*}$,
$f(x)=x+u^{2^{k-1}-1}x^{2^k-1}+u^{2^{k-1}}x^{2^k+1}$ is a PP over
$\mathbb{F}_{2^m}$.
\end{theorem}

\begin{proof}
Since $\textup{gcd}(2,2^m-1)=1$, we need only to show that
$h(x)=(f(x))^2=x^2+u^{2^{k}-2}x^{2^{k+1}-2}+u^{2^k}x^{2^{k+1}+2}$ is
a PP over $\mathbb{F}_{2^m}$. Let $\overline{u}=u^{2^k}$ and $y=x^{2^k}$.

First, we prove that $h(x)=0$ if and only if $x=0$. Clearly, if
$x=0$, then $h(x)=0$. Conversely, if there exists some
$x\in\mathbb{F}_{2^m}^{*}$ such that
\begin{equation}\label{eq:pp1}
u^{2}x^{4}+\overline{u}y^2+\overline{u}u^{2}x^{4}y^{2}=0,
\end{equation}
raising both sides of Eq. \eqref{eq:pp1} to the $2^k$-th power
gives us
$$\overline{u}^{2}y^{4}+u^{2}x^4+\overline{u}^{2}u^{2}x^{4}y^{4}=0.$$
Since $\textup{gcd}(2,2^m-1)=1$, we have
\begin{equation}\label{eq:pp2}
\overline{u}y^{2}+ux^2+\overline{u}ux^{2}y^{2}=0.
\end{equation}
Adding Eq. \eqref{eq:pp1} and Eq. \eqref{eq:pp2} together, we obtain
\begin{equation}\label{eq:pp3}
u^{2}x^{4}+ux^2+\overline{u}u^{2}x^{4}y^{2}+\overline{u}ux^{2}y^2=0,
\end{equation}
which can be factorized as $ux^2(1+ux^{2})^{1+2^k}=0$. It follows that $x^2=\frac{1}{u}$, i.e.,
 $x=\frac{1}{u^{2^{m-1}}}.$ But $h(\frac{1}{u^{2^{m-1}}})=\frac{1}{u}\neq0$, which is a
contradiction. Hence $h(x)=0$ if and only if $x=0$.

Next, if $h(x)$ is not a PP, then there exist $x\in
\mathbb{F}_{2^m}^{*}$ and $a\in \mathbb{F}_{2^m}^{*}$ such that
$h(x)=h(x+ax).$ Let $b=a^{2^k}.$ It is clear that $a,b\neq0,1.$
Since $h(x)=h(x+ax)$, we have
$$\frac{u^{2}x^4+\overline{u}y^2+\overline{u}u^{2}y^{2}x^4}{u^{2}x^2}=\frac{u^{2}(a^4+1)x^4+\overline{u}(b+1)^{2}y^{2}+\overline{u}u^2(b+1)^{2}y^{2}(a+1)^{4}x^{4}}{u^{2}(a+1)^{2}x^2},$$
which simplifies to
\begin{equation}\label{eq:pp4}
A_{1}x^{2}y^{2}+A_{2}y^{2}+A_{3}x^{2}=0,
\end{equation}
with $A_{1}=(a^{2}b^{2}+a^2+b^2+b)u\overline{u}$,
$A_{2}=(b^{2}+b)\overline{u}$, and $A_{3}=(b+a^2)u$.

We claim that $A_1A_2A_3\ne 0$. If $A_{1}=0$, we get
$(b+1)^{2}a^2=b(b+1).$ Thus $a^2=\frac{b}{b+1}.$ Raising both sides
to the $2^k$-th power, we get $b^2=\frac{a^2}{a^2+1}.$ Then $b^2=b$,
which leads to $b=0$ or $1$, a contradiction. So $A_{1}\neq0$. A
similar discussion shows that $A_{2},A_{3}\neq0$.

Raising both sides of Eq. \eqref{eq:pp4} to the $2^k$-th power,
we have
\begin{equation}\label{eq:pp5}
A_{1}^{2^k}x^{4}y^{2}+A_{3}^{2^k}y^{2}+A_{2}^{2^k}x^{4}=0.
\end{equation}
By Eq. \eqref{eq:pp4} and Eq. \eqref{eq:pp5}, cancelling $y^2$, we get
\begin{equation}\label{eq:pp6}
B_{1}x^{4}+B_{2}x^{2}+B_{3}=0,
\end{equation}
where $B_{1}=A_{3}A_{1}^{2^k}+A_{1}A_{2}^{2^k}=(b^{3}(a+1)^4)\overline{u}u^3\neq0,\ B_{2}=A_{2}^{2^k+1}\neq0,\ B_{3}=A_{3}^{2^k+1}\neq0.$

Substituting $x^2=\frac{B_2}{B_1}\gamma$ into Eq. \eqref{eq:pp6}, we obtain
\begin{equation}\label{eq:pp7}
\gamma^2+\gamma+D=0,
\end{equation}
where $D=\frac{B_{1}B_{3}}{B_{2}^2}=D_1+D_{1}^{2^k}$ and $D_1=\frac{A_{1}A_{3}^{2^k+1}}{A_{2}^{2^k+2}}=\frac{A_{1}B_3}{A_{2}B_{2}}$. We also have

\begin{align*}
\textup{Tr}_{m}(D_1)&=\textup{Tr}_{m}\left(\frac{A_1}{A_2}(\frac{A_3}{A_2})^{2^k+1}\right)\\
&=\textup{Tr}_{m}\left((1+a^2+\frac{a^2}{b})\frac{(a^2+b)(a+b)^2}{a^{2}(a+1)^{2}b(b+1)}\right)\\
&=\textup{Tr}_{m}\left(\frac{(a^2+b)(a+b)^2}{a^{2}(a+1)^{2}b(b+1)}+\frac{(a^2+b)(a+b)^2}{(a+1)^{2}b(b+1)}+\frac{(a^2+b)(a+b)^2}{(a+1)^{2}b^{2}(b+1)}\right)\\
&=\textup{Tr}_{m}\left(\frac{a^2}{(a+1)^{2}b(b+1)}+\frac{1}{(a+1)^{2}}+\frac{b^2}{a^{2}(a+1)^{2}(b+1)}+\frac{a^4}{(a+1)^{2}b(b+1)}+\frac{a^2}{(a+1)^{2}}+\right.\\
&{}\left.\frac{b^2}{(a+1)^{2}(b+1)}+\mathbf{}\frac{a^4}{(a+1)^{2}b^{2}(b+1)}+\frac{a^2}{b(a+1)^{2}}+\frac{b}{(a+1)^{2}(b+1)}\right)\\
&=\textup{Tr}_{m}\left(\frac{a^2}{(a+1)^{2}b(b+1)}+\frac{a^2}{b(a+1)^{2}}+\frac{b}{(a+1)^{2}(b+1)}\right)+\textup{Tr}_{m}\left(\frac{1}{(a+1)^{2}}+\frac{a^2}{(a+1)^{2}}\right)\\
&{}+\textup{Tr}_{m}\left(\frac{b^2}{a^{2}(a+1)^{2}(b+1)}+\frac{a^4}{(a+1)^{2}b(b+1)}+\frac{b^2}{(a+1)^{2}(b+1)}+\frac{a^4}{(a+1)^{2}b^{2}(b+1)}\right)\\
&=1.
\end{align*}

Raising both sides of Eq. \eqref{eq:pp7} to the $2^i$-th power,
$0\leq i\leq d-1,$ and then summing them up, we get
\begin{equation}\label{eq:pp8}
\gamma^{2^k}=\gamma+\sum_{i=0}^{k-1}(D_1+D_1^{2^k})^{2^i}=\gamma+D_1+\textup{Tr}_{m}(D_1)=\gamma+D_1+1.
\end{equation}
It follows that
\begin{equation}\label{eq:pp9}
\gamma^{2^k+1}=D_1\gamma+D.
\end{equation}
Combining Eqs. \eqref{eq:pp4}, \eqref{eq:pp8} and \eqref{eq:pp9}, we obtain
\begin{equation}\label{eq:pp10}
C_1\gamma+C_2=0,
\end{equation}
where $C_1=A_{1}A_{2}^{2^k+2}\neq0$ and $C_2=A_{2}^{2}B_{1}\neq0.$ Thus
$\gamma=\frac{C_2}{C_1}.$ Then from Eq. \eqref{eq:pp7}, we get
$$B_{1}A_{2}^{2}+A_{1}A_{2}B_{2}=A_{1}^{2}B_{3},$$
which leads to
$$b^2(b^2+a^4)=a^4(b^4+a^2).$$
Note that $\textup{gcd}(2,2^m-1)=1$ whence
$$a^2b^2+b^2+a^2b+a^3=0.$$
Raising both sides of the above equation to the $2^k$-th power, we
get
$$a^4b^2+a^4+a^2b^2+b^3=0.$$
It follows that
$$b^2(a^4+b)=a^2(a^2+b^2)=(a^3+a^2b)(a+b)=b^2(1+a^2)(a+b).$$
Then $$b=\frac{a^3+a^2+1}{a},$$ and raising both sides of the above
equation to the $2^k$-th power, we deduce that
$$a^8+a^7+a^6+a^5+a^4+a^2+1=0.$$
Since $x^8+x^7+x^6+x^5+x^4+x^2+1$ is irreducible over
$\mathbb{F}_{2},$ by Lemma \ref{lempp1}, $x^8+x^7+x^6+x^5+x^4+x^2+1$
is irreducible over $\mathbb{F}_{2^m}.$ Hence $a\notin
\mathbb{F}_{2^m}^{*}$, which is a contradiction. This
completes the proof.
\end{proof}

\begin{theorem}
Let $m>1$ be an odd integer such that $m=2k-1.$ Then
$f(x)=x+ux^{2^k-1}+u^{2^k}x^{2^m-2^{k+1}+2}, u\in
\mathbb{F}_{2^m}^{*},$ is a PP over
$\mathbb{F}_{2^m}$.
\end{theorem}

\begin{proof}
We first prove that $f(x)=0$ if and only if $x=0$. Let
$\overline{u}=u^{2^k}$ and $y=x^{2^k}$. Clearly, if $x=0$ then $f(x)=0$.
Conversely, suppose there exists some $x\in \mathbb{F}_{2^m}^{*}$
such that
\begin{equation}\label{eq:pp11}
x^{2}y^{2}+uy^3+\overline{u}x^{4}=0.
\end{equation}
Raising both sides of the above equation to the $2^k$-th power, we
get
\begin{equation}\label{eq:pp12}
x^{4}y^{2}+u^{2}y^4+\overline{u}x^{6}=0.
\end{equation}
Multiplying both sides of Eq. \eqref{eq:pp11} by $x^2$,
we obtain
\begin{equation}\label{eq:pp13}
x^{4}y^{2}+ux^{2}y^3+\overline{u}x^{6}=0.
\end{equation}
Adding Eq. \eqref{eq:pp12} and Eq. \eqref{eq:pp13} together, we have
$$ux^{2}y^{3}+u^{2}y^4=0.$$
It follows that $x^2=uy$. So we have $x=u^{2^{k-1}+1}$, which leads
to a contradiction since $f(u^{2^{k-1}+1})=u^{2^{k-1}+1}\neq0.$ Thus
$f(x)=0$ if and only if $x=0.$

Next, let $\overline{a}=a^{2^k}.$ We will show that $f(x)=a$ has a unique nonzero solution for
each nonzero $a\in \mathbb{F}_{2^m}$. That is, for the
equation
\begin{equation}\label{eq:pp14}
x^{2}y^{2}+uy^3+\overline{u}x^{4}+axy^2=0,
\end{equation} there exists a unique  solution $x\in
\mathbb{F}_{2^m}^{*}.$
Raising both sides of Eq. \eqref{eq:pp14} to the $2^k$-th power
and multiplying Eq. \eqref{eq:pp14} by $x^2,$ we get
\begin{equation}\label{eq:pp15}
x^{4}y^{2}+\overline{u}x^6+u^{2}y^{4}+\overline{a}x^{4}y=0,
\end{equation} and
\begin{equation}\label{eq:pp16}
x^{4}y^{2}+ux^{2}y^3+\overline{u}x^{6}+ax^{3}y^2=0.
\end{equation}
Summing Eq. \eqref{eq:pp15} and Eq. \eqref{eq:pp16}, and dividing by $y$, we
have
\begin{equation}\label{eq:pp17}
u^{2}y^3+ux^{2}y^2+\overline{a}x^4+ax^3y=0.
\end{equation}
Computing Eq. \eqref{eq:pp14}$\cdot u+$Eq. \eqref{eq:pp17}, and dividing by
$x$, we obtain
\begin{equation}\label{eq:pp18}
(\overline{a}+u\overline{u})x^3+ax^{2}y+auy^2=0.
\end{equation}
Raising both sides of Eq. \eqref{eq:pp18} to the $2^k$-th
power, and then adding $\overline{a}\cdot$Eq. \eqref{eq:pp14}, we have
\begin{equation}\label{eq:pp19}
(\overline{a}+\overline{u}u^2+\overline{a}u)y+a\overline{a}x=0.
\end{equation}
Solving Eqs. \eqref{eq:pp18} and \eqref{eq:pp19}, we get
\begin{equation}\label{eq:pp20}
x=\frac{a^3\overline{a}^2}{b\overline{b}},
\end{equation}
Here, $b=a^2+\overline{u}u^2+\overline{a}u$ and $\overline{b}=b^{2^k},$ which can be directly verified to be nonzero.
  This completes the proof.
\end{proof}
\section{Differential properties of power functions}\label{section6}
In this section, we consider the differential uniformity of monomial PPs. We first recall the basic definitions.
\begin{definition}
Let $F$ be a function from $\mathbb{F}_{2^n}$ to $\mathbb{F}_{2^m}$. For any $a\in\mathbb{F}_{2^n}$, the derivative of $F$ with respect to $a$ is the function $D_{a}(F)$ from $\mathbb{F}_{2^n}$ into $\mathbb{F}_{2^m}$ defined by
$$D_{a}(F(x))=F(x+a)+F(x),\ x\in\mathbb{F}_{2^n}.$$
\end{definition}
The resistance to differential cryptanalysis is related to the following quantities.
\begin{definition}
Let $F$ be a function from $\mathbb{F}_{2^n}$ into $\mathbb{F}_{2^n}$. For any $a$ and $b$ in $\mathbb{F}_{2^n}$, we denote
$$\delta(a,b)=\sharp\{x\in\mathbb{F}_{2^n}|D_{a}(F(x))=b\}.$$
Then the differential uniformity of $F$ is
$$\delta(F)=\textup{max}_{a\neq0,b\in\mathbb{F}_{2^n}}\delta(a,b).$$
\end{definition}

\begin{remark}
In the case $F(x)=x^d$ is a monomial, for any nonzero $a\in \mathbb{F}_{2^n},$
the equation $(x+a)^d+x^d=b$ can be rewritten as $a^d\left((\frac{x}{a}+1)^d+(\frac{x}{a})^d\right)=b$.
This implies that $\delta(a,b)=\delta(1,b/a^d)$. Therefore, for a monomial function,
 the differential properties are determined by the values $\delta(1,b), b \in \mathbb{F}_{2^n}.$
  From now on, we denote the quantity $\delta(1,b)$ by $\delta(b)$ for monomial functions.
\end{remark}

 The following two lemmas can be found in \cite{BRS}, which will be used later.
\begin{lemma}{\rm\cite{BRS}}\label{lem1}
For a positive integer $m$ and $a,b \in \mathbb{F}_{2^m},a\neq0$,
the quadratic equation $x^2+ax+b=0$ has solutions in
$\mathbb{F}_{2^m}$ if and only if $\textup{Tr}_{m}(\frac{b}{a^2})=0$.
\end{lemma}

\begin{lemma}{\rm\cite{BRS}}\label{lem2}
For a positive integer $m$ and $a \in \mathbb{F}_{2^m}^{*}$, the
cubic equation $x^3+x+a=0$ has
\begin{enumerate}
  \item[(1)] a unique solution in $\mathbb{F}_{2^m}$ if and only if
$\textup{Tr}_{m}(a^{-1}+1)=1$;
  \item[(2)] three distinct solutions in $\mathbb{F}_{2^m}$ if and only if
$p_{m}(a)=0,$ where the polynomial $p_{m}(x)$ is recursively defined
by the equations
$p_{1}(x)=p_{2}(x)=x,\ p_{k}(x)=p_{k-1}(x)+x^{2^{k-3}}p_{k-2}(x)$ for $k\geq 3;$
  \item[(3)] no solution in $\mathbb{F}_{2^m}$, otherwise.
\end{enumerate}
\end{lemma}

As a preparation, we have the following lemma.

\begin{lemma}\label{lem3}
Let $n=2m$ with $m$ odd, $b \in
\mathbb{F}_{2^n}\setminus\mathbb{F}_{2^m}$ and $y \in
\mathbb{F}_{2^m}\setminus\mathbb{F}_{2}.$ Then the number of
solutions of the equation $x^4+y^2(x^2+x+1)+x+1+b=0$ is $0\textup{ or }
4.$ Moreover, if $x_0$ is a solution, then the other three solutions
are given by $x_0+1,\ x_1$ and $x_1+1,$ where $x_1$ satisfies
$x_{1}^2+x_1=x_{0}^2+x_0+1+y^2.$
\end{lemma}
\begin{proof}
If $x_0$ is a solution of equation $x^4+y^2(x^2+x+1)+x+1+b=0$, then
$x_0+1$ is also a solution. So we have
$$x^4+y^2(x^2+x+1)+x+1+b=(x^2+x+x_{0}^2+x_0)(x^2+x+x_{0}^2+x_0+1+y^2).$$
Since $\textup{Tr}_{n}(x_{0}^2+x_0+1+y^2)=0$, by
Lemma~\ref{lem1} the equation $x^2+x+x_{0}^2+x_0+1+y^2=0$ also has $2$
solutions and these two solutions are different from $x_0$ and
$x_0+1$. Hence the number of solutions of the equation
$x^4+y^2(x^2+x+1)+x+1+b=0$ is $0\textup{ or } 4.$ The second part is
obvious.
\end{proof}

We now state our result.
\begin{theorem}
Let $n=2m$ with $m$ odd and $d=2^{m+1}+3.$ Then $F_{d}:x\rightarrow
x^d$ is a permutation over $\mathbb{F}_{2^n}$ with $\delta(F_d)\leq
10$. Moreover, for $b\in\mathbb{F}_{2^{m}}$, we have $\delta(b)\in\{0,4\}$.
\end{theorem}
\begin{proof} It can be verified that $\textup{gcd}(d,2^{n}-1)=1$, so $F_{d}$ is a permutation. For each $x \in \mathbb{F}_{2^n},$ let $\overline{x}:=x^{2^m}.$ It is clear that $x+\overline{x} \in \mathbb{F}_{2^m}$ and $x\overline{x} \in \mathbb{F}_{2^m}.$ We can then verify that $$D_{1}(F_{d}(x))=(x+1)^d+x^d=(\overline{x}^2+x)(x^2+x+1)+1=(\overline{x}x)^2+((\overline{x}+x)^2+1)(x+1).$$
Then it suffices to show that for any $b \in \mathbb{F}_{2^n}$, the equation
 $D_{1}(F_{d}(x))=b$ has at most $10$ solutions.

Assume that
\begin{equation}\label{6.1}
(\overline{x}x)^2+((\overline{x}+x)^2+1)(x+1)=b.
\end{equation}
Raising both sides of Eq. \eqref{6.1} to the $2^m$-th power, we get
\begin{equation}\label{6.2}
(\overline{x}x)^2+((\overline{x}+x)^2+1)(\overline{x}+1)=\overline{b}.
\end{equation}
Adding Eq. \eqref{6.1} and Eq. \eqref{6.2} together, we have
\begin{equation}\label{6.3}
((\overline{x}+x)^2+1)(\overline{x}+x)=b+\overline{b}.
\end{equation}
Setting $y:=\overline{x}+x\in \mathbb{F}_{2^m}$ and $a:=b+\overline{b}\in \mathbb{F}_{2^m},$ we obtain
\begin{equation}\label{6.4}
y^3+y+a=0.
\end{equation}
Replacing $\overline{x}=y+x$ into Eq. \eqref{6.1}, we have
\begin{equation}\label{6.5}
x^4+y^2(x^2+x+1)+x+1+b=0.
\end{equation}
Therefore, $x$ is a solution of Eq. \eqref{6.1} if and only if it is a solution to following equations
\begin{equation}\label{eqs1}
\left\{
\begin{array}{rcl}
x^4+y^2(x^2+x+1)+x+1+b=0, \\[1.0ex]
y^3+y+a=0, \\[1.0ex]
\overline{x}+x=y.
\end{array}\right.
\end{equation}

Hence $\delta(F_d)\leq 12$. Below, we consider the two cases where $a=0$ and $a\neq0$.

{\bf Case 1: $a=0$}

It is easy to see that $b \in \mathbb{F}_{2^m}$ and $0,\ 1$ are solutions of Eq. \eqref{6.4}. We consider these two cases separately.
\begin{enumerate}
  \item[(1)] If $y=0$, then $x\in \mathbb{F}_{2^m}.$ Thus Eq. \eqref{eqs1} will become $x^4+x+1=b,$ which has either $0$ solution or $2$ solutions. And it has $2$ solutions if and only if $\textup{Tr}_{m}(b)=1$.
  \item[(2)] If $y=1$, then Eq. \eqref{eqs1} becomes
  \begin{equation}\label{eqs2}
  \left\{\begin{array}{rcl}
  x^4+x^2+b=0, \\[1.0ex]
  \overline{x}+x=1.
  \end{array}\right.
  \end{equation}
  Since $\textup{gcd}(2,2^n-1)=1$, equation $x^4+x^2+b=0$ is equivalent to $x^2+x+b^{2^{m-1}}=0$. Clearly it has $2$ solutions by Lemma~\ref{lem1}. Since $1=\overline{x}+x=\sum_{i=0}^{m-1}(x^2+x)^{2^i}=\textup{Tr}_{m}(b^{2^{m-1}})$ we get that Eq. \eqref{eqs2} has $2$ solutions if and only if $\textup{Tr}_{m}(b)=1$.
\end{enumerate}

Note that, for both cases, Eq. \eqref{eqs1} has two solutions if and only if $\textup{Tr}_{m}(b)=1$. Therefore, $\delta(b) \in \{0,4\}$.

{\bf Case 2: $a\neq0$}

In this case it is obvious that $b \not\in \mathbb{F}_{2^m}$ and $y \not\in
\mathbb{F}_{2}$.

By Lemma~\ref{lem2}, Eq. \eqref{6.4} has $0$, $1$ or $3$ solutions. We consider these cases separately.

\begin{enumerate}
  \item[(1)] If Eq. \eqref{6.4} has no solution, then $\delta(b)=0.$
  \item[(2)] If Eq. \eqref{6.4} has one solution, say $y_0$, then by Lemma~\ref{lem3}, Eq. \eqref{6.5} has 0 or 4 solutions given by $x_{11},x_{11}+1,x_{21}$ and $x_{21}+1$. However, we need $x_{i1}+\overline{x_{i1}}=y_0$ holds for $i=1,2$. Thus $\delta(b)\in \{0,2,4\}.$
  \item[(3)] If Eq. \eqref{6.4} has three solutions, denoted by $y_1,y_2$ and $y_3$. Then we have $y_1+y_2+y_3=0.$ For each $y_i,\ 1\leq i\leq 3,$ by Lemma~\ref{lem3}, there are $0$ or $4$ solutions for Eq. \eqref{6.5}. So the total number of solutions for $x$ is $0,\ 4,\ 8$ or $12.$
  \begin{enumerate}
      \item[(i)] If the number of solutions for $x$ is $0,\ 4$ or $\ 8$, then the number of solutions to Eq. \eqref{eqs1} is at most 8. Thus $\delta(b)\in \{0,2,4,6,8\}.$
      \item[(ii)] If the number of solutions for $x$ is $12$, that is, for each $y_i,\ 1\leq i\leq 3,$ there are $4$ solutions of Eq. \eqref{6.5}. Let the $12$ solutions be $\{x_{ij},x_{ij}+1 | i=1,2,3;\ j=1,2\},$ with $y_{i}$ corresponding to $x_{i1},x_{i1}+1,x_{i2}$ and $x_{i2}+1$. If $x_{i1},x_{i1}+1,x_{i2}$ and $x_{i2}+1$ are exactly the solutions of Eqs. \eqref{eqs1}, we can easily get
          \begin{equation}\label{eq2}
          \left\{
          \begin{array}{lcl}
          x_{i1}+x_{i2}+(x_{i1}+x_{i2})^{2}=1+y_{i}^{2}, \\
          x_{i1}+x_{i2}\in \mathbb{F}_{2^m},
          \end{array}\right.
          \end{equation}
          which means that the equation $t^2+t+1+y_{i}^{2}=0$ has 2 solutions over $\mathbb{F}_{2^m}.$ By Lemma~\ref{lem1}, we have $\textup{Tr}_{m}(1+y_{i}^2)=0$ so that $\textup{Tr}_{m}(y_{i})=1.$ Therefore, if $\delta(b)=12$, then $\textup{Tr}_{m}(y_{1})=\textup{Tr}_{m}(y_{2})=\textup{Tr}_{m}(y_{3})=1.$ However, $1=\textup{Tr}_{m}(y_{1})+\textup{Tr}_{m}(y_{2})+\textup{Tr}_{m}(y_{3})=\textup{Tr}_{m}(y_{1}+y_{2}+y_{3})=0$, which is a contradiction. So $\delta(b)\leq 10.$
  \end{enumerate}
\end{enumerate}
Thus we obtain $\delta(F_d)\leq 10$.
\end{proof}

\begin{remark}
For $d=2^m+2^{(m+1)/2}+1,$ we can obtain $\delta(F_d)\leq 10$ in a similar way. In fact, with $m=2r-1$
 and $a:=b+\overline{b}$, Eq. \eqref{eqs1} now take the following form
\begin{equation}\label{eqs3}
\left\{
\begin{array}{rcl}
x^4+(ay+1)x^2+ayx+(y^2+1)\overline{b}^{2^r}+b\overline{b}=0, \\[1.0ex]
(y+1)x^{2^r}+x^2+yx+y+1+b=0, \\[1.0ex]
y^3+(a+1)y^2+a^{2^r}y+a^{2^r}=0, \\[1.0ex]
\overline{x}+x=y.
\end{array}\right.
\end{equation}
Hence $\delta(F_d)\leq 12.$ If $b\in\mathbb{F}_{2^{m}}$, we can
easily get $\delta(b)\in\{0,4\}$. If $\delta(b)=12$ for some
$b\in\mathbb{F}_{2^{n}}\backslash\mathbb{F}_{2^{m}},$ we have
$\textup{Tr}_{m}(ay_{i})=1,$ and
$1=\textup{Tr}_{m}(ay_{1})+\textup{Tr}_{m}(ay_{2})+\textup{Tr}_{m}(ay_{3})=\textup{Tr}_{m}(a(a+1))=0$,
which is a contradiction.
\end{remark}

 Hence we have the following result.
\begin{theorem}
Let $n=2m$ with $m$ odd and $d=2^m+2^{(m+1)/2}+1.$ Then $F_{d}:x\rightarrow
x^d$ is a permutation over $\mathbb{F}_{2^n}$ with $\delta(F_d)\leq
10$. Moreover, for $b\in\mathbb{F}_{2^{m}}$, we have $\delta(b)\in\{0,4\}$.
\end{theorem}

\begin{remark}
Here we give a concrete example to illustrate the idea of the proof. Let $w$ be a primitive
element of $\mathbb{F}_{2^n}$, $n=10$, $d=67$,
$b=w^{27}$ and $a=b+\overline{b}$. Then Eq. \eqref{6.4} and
Eq. \eqref{6.5} become $y^3+y+a=0$ and $x^4+y^2(x^2+x+1)+x+1+b=0$ respectively.

\begin{tabular}{c|c|c|c}
\hline
Solutions of   &                        &  Solutions of     &  Solutions of           \\
$y^3+y+a=0$    &$\textup{Tr}_{m}(y_{i})$ &$x^4+y^2(x^2+x+1)+x+1+b=0$   & $D_{1}(F_{d}(x))=b$\\\hline
 $y_1=w^{330}$  &$1$                      &$\{w^{672},w^{1019};w^{619},w^{975}\}$ &      $w^{226},w^{633},$               \\  \cline{1-3}
 $y_2=w^{363}$  &$1$                      &$\{w^{226},w^{633};w^{586},w^{903}\}$  &         $w^{586},w^{903},$              \\  \cline{1-3}
 $y_3=w^{924}$  &$0$                      &$\{w^{129},w^{340};w^{774},w^{883}\}$  &      $w^{129},w^{340}$                \\ \hline
\end{tabular}\\

In the above table, for a fixed element $b$, we obtain $3$ solutions
of Eq. \eqref{6.4}, denoted by $y_1,\ y_2$ and $\ y_3.$ For
each $y_i,$ by Eq. \eqref{6.5} we get $4$ solutions. We need to
determine whether they are satisfying $x+\overline{x}=y_i$. Since
$\textup{Tr}_{m}(y_{3})=0,$ there are at least $2$ solutions which
are not satisfying $x+\overline{x}=y_3$ (in the above example, $w^{774}$
and $w^{883}$ are not). However, for each $y_i,\ i=1,2,$ we can not
determine whether the $4$ solutions are satisfying $x+\overline{x}=y_i$
since $\textup{Tr}_{m}(y_{1})=\textup{Tr}_{m}(y_{2})=1$ (in the
above example, the four solutions corresponding to $y_1$ are not,
while the four solutions corresponding to $y_2$ are). Therefore, in
our proof we only can get $\delta(b)\leq 10,$ but in fact
$\delta(b)=6$ for this example. Thus, we need to find more detailed
conditions to characterize the solutions of the equation.
\end{remark}

\section{Conclusion}\label{conclusion}
Permutation and CPPs have important applications in cryptography. This paper demonstrates some new results on permutation and CPPs. First, by using the AGW Criterion, we proved a conjecture proposed by Wu et al. \cite{WLHZ}. Then we give three other new classes of monomial CPPs over finite fields and the main tool is additive characters over the underlying finite fields. Moreover, a class of trinomial CPPs and two classes of trinomial PPs are also presented in this paper. Finally, for $d=2^{m+1}+3$ or $2^{m}+2^{\frac{m+1}{2}}+1$, Blondeau et al. \cite{BCC} conjectured that $x^{d}$ is differentially $8$-uniform over $\mathbb{F}_{2^{n}}$, where $n=2m$. We make some progress towards this conjecture and prove that the differential uniformity of $x^{d}$ is at most $10$. It seems not easy to exclude the possibility that $\delta(b)=10$ for some $b\in\mathbb{F}_{2^{n}}\backslash\mathbb{F}_{2^{m}}$. We look forward to seeing further progress on this conjecture.

\section*{Acknowledgements}
The authors express their gratitude to the anonymous reviewers for their detailed and constructive comments which are very helpful to the improvement of the presentation of this
paper. The research of T. Feng was supported by Fundamental Research Fund for the Central Universities of China, the National Natural Science Foundation
of China under Grant No.~11201418 and  Grant No.~11422112, and the Research Fund for Doctoral Programs from the Ministry of Education of
China under Grant 20120101120089. The research of G. Ge was supported by the National Natural Science Foundation of China under Grant Nos. 11431003 and 61571310.

\end{document}